\documentclass[12pt]{amsart}
\usepackage{amssymb, latexsym}

\def\nr#1{\buildrel{#1}\over=}

\def\be{\begin{enumerate}}

\def\ee{\end{enumerate}}

\theoremstyle{plain}

\newtheorem{thm}{Theorem}[section]

\newtheorem{lemma}[thm]{Lemma}

\theoremstyle{definition}

\newtheorem{de}[thm]{Definition}

\newtheorem{example}[thm]{Example}

\numberwithin{equation}{section}

\title[Sharp and principal...]{Sharp and principal elements\\ in effect algebras}

\begin{document}

\author[G. Bi\'nczak]{G. Bi\'nczak$^1$}

\address{$^1$ Faculty of Mathematics and Information Sciences\\
Warsaw University of Technology\\
00-662 Warsaw, Poland}

\author[J. Kaleta]{J. Kaleta$^2$}

\address{$^2$ Department of Applied Mathematics\\
 Warsaw University of Agriculture\\02-787 Warsaw, Poland}

\email{$^1$ binczak@mini.pw.edu.pl, $^2$joanna\_kaleta@sggw.pl}



\begin{abstract}
In this paper we characterize the effect algebras whose sharp and principal elements coincide. 
We present partially solution to the problems: when the set of sharp (principal) elements is closed under orthosum. We also give examples of two non-isomorphic effect algebras having the same universum, partial order and orthosupplementation.
\end{abstract}

\maketitle

\section{\bf Introduction}
Effect algebras have been introduced by Foulis and Bennet in 1994 (see \cite{FB94}) for the study of foundations of quantum mechanics (see \cite{DP00}). Independtly, Chovanec and K\^opka introduced an essentially equivalent structure called $D$-{\em poset} 
(see \cite{KC94}). Another equivalent structure was introduced by Giuntini and Greuling in \cite{GG94}).

The most important example of an effect algebra is $(E(H),0,I,\oplus)$, where $H$ is a Hilbert space and $E(H)$ consists of all self-adjoint operators $A$ on $H$ such that $0\leq A\leq I$. For $A,B\in E(H)$, $A\oplus B$ is defined if and only if $A+B\leq I$ and then $A\oplus B=A+B$. Elements of $E(H)$ are called {\em effects} and they play an important role in the theory of quantum measurements (\cite{BLM91},\cite{BGL95}).

A quantum effect may be treated as two-valued (it means $0$ or $1$) quantum  measurement that may be unsharp (fuzzy).  If there exist some pairs of effects $a,b$ which posses an orthosum $a\oplus b$ then this orthosum correspond to a  parallel measurement of two effects. 

In this paper we solved the following  Open Problem:  Characterize the effect algebras whose sharp  and principal elements coincide (see~\cite{G96}). So far it was known (see Theorem 3.16 in \cite{BF95}) that if effect algebra $E$ is lattice-ordered then $e\in E$ is principal iff $e\land e'=0$. It also was known that in every effect algebra any principal element is sharp (see Lemma 3.3 in \cite{GFP95}).

\begin{de}
In \cite{FB94} an {\em effect algebra} is defined to be an algebraic system $(E,0,1,\oplus)$  consisting of a set $E$, two special elements $0,1\in E$ called  the {\em zero} and the {\em unit}, and a partially defined binary operation $\oplus$ on $E$ that satisfies  the following conditions for all $p,q,r\in E$:
\be
\item{} [Commutative Law] If $p\oplus q$ is defined, then $q\oplus p$ is defined and $p\oplus q=q\oplus p$.
\item {}[Associative Law] If $q\oplus r$ is defined and $p\oplus(q\oplus r)$ is defined, then $p\oplus q$ is defined, $(p\oplus q)\oplus r$ is defined, and $p\oplus(q\oplus r)=(p\oplus q)\oplus r$.
\item {}[Orthosupplementation Law] For every $p\in E$ there exists a unique $q\in E$ such that $p\oplus q$ is defined and $p\oplus q=1$.
\item {}[Zero-unit Law] If $1\oplus p$ is defined, then $p=0$.
\ee
\end{de}

For simplicity , we often refer to $E$, rather than to $(E,0,1,\oplus)$, as being an effect algebra.

If $p,q\in E$, we say that $p$ and $q$ are orthogonal and write $p\perp q$ iff $p\oplus q$ is defined in $E$. If $p,q\in E$ and $p\oplus q=1$, we call $q$ the {\em orthosupplement} of $p$ and write $p'=q$.

It is shown in \cite{FB94} that the relation $\leq$ defined for $p,q\in E$ by $p\leq q$ iff $\exists r\in E$ with $p\oplus r=q$ is a partial order on $E$ and $0\leq p\leq 1$ holds for all $p\in E$. It is also shown that the mapping $p\mapsto p'$ is an order-reversing involution and that $q\perp p$ iff $q\leq p'$. Furtheremore, $E$ satisfies the following {\em cancellation law}: If $p\oplus q\leq r\oplus q$, then $p\leq r$.

An element $a\in E$ is {\em sharp} if the greatest lower bound  of the set $\{a,a'\}$ equals $0$ (i.e. $a\land a'=0$). We denote the set of sharp elements of $E$ by $S_E$.

An element $a\in E$ is said to be {\em principal} iff  for all $p,q\in E$, $p\perp q$ and $p,q\leq a\Rightarrow p\oplus q\leq a$. We denote the set of principal elements of $E$ by $P_E$.

\begin{de}
For effect algebras $E_1,E_2$ a mapping $\phi\colon E_1\to E_2$ is said to be an {\em isomorphism} if  $\phi$ is a bijection, 
$a\perp b\iff \phi(a)\perp\phi(b)$, $\phi(1)=1$ and $\phi(a\oplus b)=\phi(a)\oplus \phi(b)$.
\end{de}

Let us observe that if $\phi\colon E_1\to E_2$ is an isomorphism then $\phi(0)=0$, because $\phi(0)\oplus 0=\phi(0)=\phi(0\oplus 0)=\phi(0)\oplus\phi(0)$ so by cancellation law $0=\phi(0)$.

\begin{de}
A quasigroup $(Q,\cdot)$ consists of a non-empty set $Q$ equipped with a  one binary operation $\cdot$ such that if any two of $a, b, c$ are given elements of a quasigroup, $ab = c$ determines the third uniquely as an element of the quasigroup.

Moreover  if  $a\cdot b=c\iff  c\cdot a=b$ then $Q$ is called {\em semisymmetric} (see \cite{S97}).
Commutative semisymmetric quasigroups are called {\em totally symmetric} (see \cite{E65}).
\end{de}

\section{Main Theorem}
In order to characterize principal elements we need the following known Theorem:

\begin{thm} \cite[Theorem 3.5]{GFP95}\label{3.5}
If $p,q\in E$, $p\perp q$, and $p\lor q$ exists in $E$, then $p\land q$ exists in $E$, $p\land q\leq(p\lor q)'\leq(p\land q)'$ and
$p\oplus q=(p\land q)\oplus (p\lor q)$.
\end{thm}

\begin{lemma}\label{lem4}
Let $(E,0,1,\oplus)$ be an effect algebra. If $x\in P_E$, $t\in E$ and $t\leq x$ then there exists $t\lor x'$ in $E$ and
$$t\lor x'=t\oplus x'.$$
\end{lemma}
\begin{proof}
Suppose that $x\in P_E$.
Let $t\in E$ and $t\leq x$ hence $t\perp x'$. We show that $t\oplus x'$ is the join of $t$ and $x'$. 

Obviously $t\leq t\oplus x'$ and $x'\leq t\oplus x'$. Suppose that $u\in E$, $t\leq u$ and $x'\leq u$ then
$$t\perp u'\eqno(1)$$ and 
$$u'\leq x\quad t\leq x.\eqno(2)$$

Now $(1)$ and $(2)$ implies $t\oplus u'\leq x$ since $x\in P_E$. Hence
$x'\perp (t\oplus u')$ and by associativity $x'\perp t$ and $(x'\oplus t)\perp u'$ thus $t\oplus x'\leq u$ so $t\oplus x'$ is the smallest upper bound of the set $\{t,x'\}$ thus $t\oplus x'=t\lor x'$.
\end{proof}

It turns out that under some conditions every effect algebra satisfies the de Morgan's law.
\begin{lemma}\label{lemdemor}
Let $(E,0,1,\oplus)$ be an effect algebra. If $x,y\in E$ and there exists $x\lor y$ in E then there exists $x'\land y'$ in $E$ and
$$x'\land y'=(x\lor y)'$$
\end{lemma}
\begin{proof}
Let $x,y\in E$ and suppose that there exists $x\lor y$ in $E$.

We show that
$$
x'\land y'=(x\lor y)'\eqno(3)
$$

We show that $(x\lor y)'$ is a lower bound of the set $\{x',y'\}$:
$x\leq x\lor y\Rightarrow x'\geq (x\lor y)'$ and $y\leq x\lor y\Rightarrow y'\geq (x\lor y)'$.

If $v$ is a lower bound of $\{x',y'\}$ then $x'\geq v$, $y'\geq  v$ thus $x\leq v'$ and $y\leq v'$ hence
$x\lor y\leq v'$  and $(x\lor y)'\geq v$ and it implies that $(x\lor y)'$ is the greatest lower bound of the set $\{x',y'\}$ so $(3)$ is satisfied.

\end{proof}
Similarly we obtain
\begin{lemma}\label{lemdemor2}
Let $(E,0,1,\oplus)$ be an effect algebra. If $x,y\in E$ and there exists $x\land y$ in E then there exists $x'\lor y'$ in $E$ and
$$x'\lor y'=(x\land y)'$$
\end{lemma}

\begin{thm}\label{thm:pe}
Let $(E,0,1,\oplus)$ be an effect algebra. Then
$$
P_E=\{x\in E\colon x\in S_E\hbox{ and}\;\forall_{t\in E}t\leq x\Rightarrow t\lor x'\;\hbox{exists in}\;E\}
$$
\end{thm}
\begin{proof}
Suppose that $x\in P_E$ then $x\in S_E$ (see Lemma 3.3 in \cite{GFP95}).
Let $t\in E$ and $t\leq x$ . Then  there exists $t\lor x'$ in $E$ by Lemma \ref{lem4}.

Suppose that $x\in S_E$ and 
$$\forall_{t\in E}t\leq x\Rightarrow t\lor x'\;\hbox{ exists in}\;E.\eqno(4)$$
 We show that $x\in P_E$.

 If $u,s\in E$,  $u\leq x$, $s\leq x$ and $u\perp s$ then
$$u\land x'=0\eqno(5)$$
 because: if $y\leq x'$ and $y\leq u\leq x$ then $y=0$ since $x\land x'=0$.

Moreover $u\leq x$ so $u\lor x'$ exists by $(4)$. By Theorem \ref{3.5}

$$u\oplus x'=(u\land x')\oplus (u\lor x')\nr{(5)}u\lor x'\eqno(6)$$

By Lemma \ref{lemdemor} we have

$$
u'\land x=(u\lor x')'.\eqno(7)
$$

Moreover $s\leq u'$ (since $u\perp s$) and $s\leq x$ so
$s\leq u'\land x$. Hence by $(6)$ and $(7)$ we have

$$
s\leq u'\land x=(u\lor x')'=(u\oplus x')'
$$
so $s\perp(u\oplus x')$ and by associativity $s\oplus u\perp x'$ hence $s\oplus u\leq x$ and $x\in P_E$.

\end{proof}

In the following theorem we prove that in every effect algebra $E$ sharp and principal elements coincide if and only if
there exists in $E$ join of every two orthogonal elements such that one of them is sharp.

\begin{thm}
Let $(E,0,1,\oplus)$ be an effect algebra. Then $S_E=P_E$ if and only if 
$$
\forall_{t,x\in E}(\;t\perp x' \;\hbox{and}\; x\land x'=0)\Rightarrow t\lor x' \;\hbox{exists in}\;E\eqno(8)
$$
\end{thm}
\begin{proof}
Suppose that $S_E=P_E$. We show that $(8)$ is satisfied. 

Let $x,t\in E$, $t\perp x'$ and $x\land x'=0$. Then $t\leq x$, $x\in P_E$ and by Theorem \ref{thm:pe} we know that $t\lor x' $ exists in $E$.

Suppose that condition $(8)$ is fullfilled. Obviously $P_E\subseteq S_E$ (see Lemma 3.3 in \cite{GFP95}).

Now our task is to show that $S_E\subseteq P_E$. Let $x\in S_E$. If $t\in E$,  $t\leq x$ then $t\perp x'$ and by condition (8) 
$ t\lor x' $ exists in $E$ hence $x\in P_E$ by Theorem \ref{thm:pe}.  Thus $S_E\subseteq P_E$.
\end{proof}
                                                              
\begin{lemma}\label{lem3}                                                                                                                                    
Let $(E,0,1,\oplus)$ be an effect algebra. If $P_E$ is closed under $\oplus$ (that is, if $x,y\in P_E$ and $x\perp y$, then $x\oplus y\in P_E$) then
$$
\forall_{x,y\in P_E}\quad x\perp y\Rightarrow x'\land(x\oplus y)=y.
$$
\end{lemma}
\begin{proof}
Let $x,y\in P_E$ and $x\perp y$. Then $y\leq x'$ and $y\leq x\oplus y$ so $y$ is a lower bound of $x'$ and $x\oplus y$.
Let $t$ be a lower bound of $x'$ and $x\oplus y$. We show that $t\leq y$. We know that $t\leq x'$ so $t\perp x$. Moreover
$t\leq x\oplus y$ and $x\leq x\oplus y$ hence 
$$x\oplus t\leq x\oplus y$$
since $x\oplus y\in P_E$. After using cancellation law we obtain $t\leq y$. Hence $y$ is the largest lower bound of $x'$ and $x\oplus y$ so
$ x'\land(x\oplus y)=y$.
\end{proof}

\begin{lemma}\label{lem5}
Let $(E,0,1,\oplus)$ be an effect algebra. If for every $x,y\in S_E$ such that $x\perp y$ there exists $x\lor y$ in $E$ and
$$
x'\land(x\lor y)=y,\quad x\oplus y=x\lor y\eqno(9)
$$
then $S_E$ is closed under $\oplus$  (that is, if $x,y\in S_E$ and $x\perp y$, then $x\oplus y\in S_E$).
\end{lemma}
\begin{proof}
Let $x,y\in S_E$ and $x\perp y$. By (9) we have
$$
\begin{array}{rl}
0&{}=y'\land y=y'\land(x'\land(x\lor y))\nr{Lemma\;\ref{lemdemor}}(x\lor y)'\land(x\lor y)\vspace{5mm}\\
&{}=(x\oplus y)'\land (x\oplus y),
\end{array}
$$ so $x\oplus y\in S_E$ and $S_E$ is closed under $\oplus$.
\end{proof}

In the following Theorem we show that if $S_E=P_E$ then $S_E=P_E$ is closed under $\oplus$ if and only if elements in $S_E=P_E$ satisfy
the orthomodular law. It partially solves Open problems 3.2 and 3.3 in \cite{G96}.
\begin{thm}
Let $(E,0,1,\oplus)$ be an effect algebra such that $S_E=P_E$. Then $S_E=P_E$ is closed under $\oplus$ if and only if for every
$x,y\in S_E$ we have
$$
x\leq y\Rightarrow x\lor(x'\land y)=y.
$$
\end{thm}
\begin{proof}
Suppose that $S_E=P_E$ is closed under $\oplus$. Let $x,y\in S_E$ and $x\leq y$ then $x\perp y'$ and by Lemma \ref{lem3}
we have $x'\land(x\oplus y')=y'$ since $y'\in S_E$. It follows that 
$$
\begin{array}{rl}
y&{}\nr{Lemma\;\ref{lemdemor2}}x\lor(x\oplus y')'\nr{Lemma\;\ref{lem4}}x\lor(x\lor y')' \vspace{5mm}\\
&{}\nr{Lemma\;\ref{lemdemor2}}x\lor(x'\land y).
\end{array}
$$

Suppose that for every
$x,y\in S_E$ we have 
$$x\leq y\Rightarrow x\lor(x'\land y)=y\eqno(10).$$ 
We show that  for every $x,y\in S_E$ such that $x\perp y$ there exists $x\lor y$ in $E$ and
$$
x'\land(x\lor y)=y,\quad x\oplus y=x\lor y.
$$
Let  $x,y\in S_E$ and $x\perp y$. Then $x\leq y'$ and $y'\in S_E=P_E$, so $x\lor(y')'=x\lor y$ exists in $E$  and
$x\lor y=x\oplus y$ by Lemma \ref{lem4}. Moreover $x\lor(x'\land y')=y'$ by (10). Hence $x'\land(x'\land y')'=y$ by Lemma \ref{lemdemor}
and $x'\land (x\lor y)=y$ by Lemma \ref{lemdemor2}. Therefore $S_E$ is closed under $\oplus$ by Lemma \ref{lem5}.

\end{proof}

Let us observe that by Theorem \ref{thm:pe} principal elements in an effect algebra are determined by partial order $\leq$ and orthosupplementation $'$. We will see that there exist effect algebras $E_1=(E,0,1,\oplus_1)$ and $E_2=(E,0,1,\oplus_2)$ such that  orthosupplementation $'$  in $E_1$ and orthosupplementation $'$ in $E_2$ are equal and also the same is true for partial order $\leq$, but $E_1$ and $E_2$ are not isomorphic.

\begin{de}
Let $(Q,\cdot)$ be a totally symmetric quasigroup.

We define $E(Q,\cdot):=\biggl((Q\times\{0\})\cup (Q\times\{1\})\cup\{0\}\cup\{1\},0,1,\oplus\biggr)$ where 
\begin{itemize}
\item $(q_1,0)\oplus(q_2,0)=(q_1\cdot q_2,1)$ for all $q_1,q_2\in Q$,
\item $(q,0)\oplus(q,1)=(q,1)\oplus(q,0)=1$ for all $q\in Q$,
\item $0\oplus x=x\oplus 0=x$ for all $x\in (Q\times\{0\})\cup (Q\times\{1\})\cup\{0\}\cup\{1\}$.
\end{itemize}
In  the remaining cases orthosum $x\oplus y$ is not defined.
\end{de}

\begin{thm}\label{thm:q}
If $(Q,\cdot)$ is  a totally symmetric quasigroup then $E(Q,\cdot)$ is an effect algebra.
\end{thm}
\begin{proof}
The Commutative Law and Zero-unit Law are obvious. If $q\in Q$ then there exists a unique element $x=(q,1)$ such that $(q,0)\oplus x=1$ so $(q,0)'=(q,1)$. Similarly $(q,1)'=(q,0)$ so the Orthosupplementation Law is satisfied.

 It remains to show that the Associative Law
is also fullfilled. Let $x,y,z\in(Q\times\{0\})\cup (Q\times\{1\})\cup\{0\}\cup\{1\}$. If $x=0$ or $y=0$, or $z=0$ then The Asociative Law is true. If $y\oplus z$ is defined and $x\oplus(y\oplus z)$ is defined and $x,y,z\not=0$ then $x,y,z\in Q\times\{0\}$, so there exist  $p,q,r\in Q$ such that $x=(p,0)$, $y=(q,0)$, $z=(r,0)$, so $(q,0)\oplus (r,0)$ is defined and $(p,0)\oplus((q,0)\oplus(r,0))$ is defined, then 
$(p,0)\oplus(q\cdot r,1)$ is defined so $q\cdot r=p$ hence $p\cdot q =r$ thus $(p\cdot q,1)\oplus(r,0)$ is defined so
$((p,0)\oplus(q,0))\oplus (r,0)$ is defined and $(p,0)\oplus ((q,0)\oplus (r,0))=((p,0)\oplus(q,0))\oplus(r,0)=1$. Therefore 
$(x\oplus y)\oplus z$ is defined and $x\oplus (y\oplus z)=(x\oplus y)\oplus z=1$.

\end{proof}

\begin{example}
Let $Q=\{1,2,3\}$ and 
\bigskip

$\begin{array}{c|ccc}
\cdot_1&1&2&3\\\hline
1&1&3&2\\
2&3&2&1\\
3&2&1&3
\end{array}\quad$
$\begin{array}{c|ccc}
\cdot_2&1&2&3\\\hline
1&2&1&3\\
2&1&3&2\\
3&3&2&1
\end{array}$ 
\bigskip

then  $E(Q,\cdot_1)$ and $(Q,\cdot_2)$ are  totally symmetric quasigroups (see Example 2 and 3 in \cite{E65}). Then by Theorem \ref{thm:q} $E(Q,\cdot_1)$ and
$E(Q,\cdot_2)$ are effect algebras with the following $\oplus$ tables. In this tables we do not include $0$ and $1$, since they have trivial sums and a dash means that the corresponding $\oplus$ is not defined:
\bigskip

$\begin{array}{c|cccccc}
\oplus_1&a_1&a_2&a_3&a_1'&a_2'&a_3'\\\hline
a_1&a_1&a_3&a_2&1&-&-\\
a_2&a_3&a_2&a_1&-&1&-\\
a_3&a_2&a_1&a_3&-&-&1\\
a_1'&1&-&-&-&-&-\\
a_2'&-&1&-&-&-&-\\
a_3'&-&-&1&-&-&-
\end{array}$
\bigskip

$\begin{array}{c|cccccc}
\oplus_2&a_1&a_2&a_3&a_1'&a_2'&a_3'\\\hline
a_1&a_2&a_1&a_3&1&-&-\\
a_2&a_1&a_3&a_2&-&1&-\\
a_3&a_3&a_2&a_1&-&-&1\\
a_1'&1&-&-&-&-&-\\
a_2'&-&1&-&-&-&-\\
a_3'&-&-&1&-&-&-
\end{array}$
\bigskip

where $a_i=(i,0)$ and $a_i'=(i,1)$ for $i=1,2,3$. In effect algebras $E(Q,\cdot_1)$ and $E(Q,\cdot_2)$ partial order $\leq$ is the same:
$a_1,a_2,a_3$ are minimal nonzero elements, $a_1',a_2',a_3'$ are maximal elements not equal to $1$, moreover $a_i\leq a_j'$ for all
$i,j\in\{1,2,3\}$. Obviously  orthosupplementation $'$ is the same in both effect algebras mentioned above. But $E(Q,\cdot_1)$ and $E(Q,\cdot_2)$ are not isomorphic:

Suppose that a mapping $\phi\colon (Q\times\{0\})\cup (Q\times\{1\})\cup\{0\}\cup\{1\}\to (Q\times\{0\})\cup (Q\times\{1\})\cup\{0\}\cup\{1\}$ is an isomorphism of $E(Q,\cdot_1)$ onto $E(Q,\cdot _2)$. Then
$$\phi(a_1)\oplus_2\phi(a_1)=\phi(a_1\oplus_1 a_1)=\phi(a_1)=\phi(a_1)\oplus_20$$

so $\phi(a_1)=0$, but $\phi(0)=0$ hence $a_1=0$ and we obtain a contradiction.

So in fact effect algebras $E(Q,\cdot_1)$ and $E(Q,\cdot_2)$ are not isomorphic and it follows that in some effect algebras partial order
$\leq$, and orthosupplementation $'$ do not determine $\oplus$.

\end{example}

\end{document}